\documentclass[sn-mathphys-num]{sn-jnl}

\usepackage{graphicx}%
\usepackage{multirow}%
\usepackage{amsmath,amssymb,amsfonts}%
\usepackage{amsthm}%
\usepackage{mathrsfs}%
\usepackage[title]{appendix}%
\usepackage{xcolor}%
\usepackage{textcomp}%
\usepackage{manyfoot}%
\usepackage{booktabs}%
\usepackage{algorithm}%
\usepackage{algorithmicx}%
\usepackage{algpseudocode}%
\usepackage{listings}%

\usepackage[utf8]{inputenc}
\usepackage[scaled=0.8]{DejaVuSansMono}
\usepackage{tikz}
\usepackage[T1]{fontenc}

\usetikzlibrary{positioning, fit, shapes.misc}

\newcommand{\N}{\mathbb{N}}
\newcommand{\is}{\,\operatorname{\bf is}\,}

\DeclareUnicodeCharacter{2203}{$\exists$}
\DeclareUnicodeCharacter{2227}{$\land$}

\usepackage{color}
\definecolor{keywordcolor}{rgb}{0.7, 0.1, 0.1}   
\definecolor{commentcolor}{rgb}{0.4, 0.4, 0.4}   
\definecolor{symbolcolor}{rgb}{0.0, 0.1, 0.6}    
\definecolor{sortcolor}{rgb}{0.1, 0.5, 0.1}      

\usepackage{listings}
\lstMakeShortInline"

\newtheorem{theorem}{Theorem}

\newtheorem{remark}{Remark}%

\begin{document}

\title[Reimplementing Mizar]{Reimplementing Mizar in Rust}


\author*{\fnm{Mario} \sur{Carneiro}}\email{mcarneir@andrew.cmu.edu}

\affil{\orgdiv{Hoskinson Center for Formal Mathematics}, \orgname{Carnegie Mellon University}, \orgaddress{\street{5000 Forbes Avenue}, \city{Pittsburgh}, \state{PA} \postcode{15213}, \country{USA}}}

\abstract{
  This paper describes a new open-source proof processing tool, \texttt{mizar-rs}, a wholesale reimplementation of the Mizar proof system, written in Rust. In particular, the main Mizar processing tool \texttt{verifier} is implemented, which can be used to verify the Mizar Mathematical Library (MML). This is to our knowledge the first and only external implementation of the system. We have used this to verify the entire MML in 11.8 minutes on 8 cores, a 4.8$\times$ speedup over the original Pascal implementation. Since Mizar is not designed to have a small trusted core, checking Mizar proofs entails following Mizar closely, so our ability to detect bugs is limited. Nevertheless, we were able to find multiple memory errors, four soundness bugs in the original (which were not being exploited in MML), in addition to one non-critical bug which was being exploited in 46 different MML articles. We hope to use this checker as a base for proof export tooling, as well as revitalizing development of the language.}

\keywords{Mizar, proof checker, software, Rust}



\maketitle

\section{Introduction}

The Mizar language \cite{bancerek2015mizar} is a proof language designed by Andrzej Trybulec in 1973 for writing and checking proofs in a block structured natural deduction style. The Mizar project more broadly has been devoted to the development of the language and tooling, in addition to the Mizar Mathematical Library (MML) \cite{bancerek2018role}, a compendium of ``articles'' on a variety of mathematical topics written in the Mizar language. The MML is one of the largest and oldest formal mathematical libraries in existence, containing (at time of writing) 1434 articles and over 65,000 theorems.

The ``Mizar system'' is a collection of tools for manipulating Mizar articles, used by authors to develop and check articles for correctness, and maintained by the Association of Mizar Users (SUM). Arguably the most important tool in the toolbox is \texttt{verifier}, which reads a Mizar article and checks it for logical correctness. The starting point for this work is the goal of extracting formal proofs from the \texttt{verifier} which can be checked by a tool not directly connected to Mizar. This turned out to be quite challenging, and this paper will explain how we achieved a slightly different goal -- to build \texttt{mizar-rs}, a Rust program which checks proofs in the same manner as the \texttt{verifier}, with an eye for proof generation.

\subsection*{Why hasn't this been done before?}

The first challenge is that there is no ``spec'' for the Mizar language: a Mizar proof is one that \texttt{verifier} accepts. Some of it is ``obvious'' things which users can figure out by playing with the language, and there are reference works at varying levels of detail \cite{grabowski2010mizar, muzalewski1993outline}, but if we want to check the entire MML then approximations aren't going to cut it. There are some writings online which give the general sense of the algorithm, but there is nowhere to look for details except the source code.

This leads to the second challenge, which is that the source code for the Mizar system is not publicly available, or at least was not while this project was under development.\footnote{A happy after-effect of this project is that the Mizar source code itself has been made public! It is now available at \url{https://github.com/MizarProject/system}.} It was only made available to members of the SUM. Luckily, it is possible to obtain a membership, but the price of admission is that one has to write a Mizar article first. Thus, for this project the author contributed an article regarding the divergence of the sum of prime reciprocals \cite{PRIMRECI}.

The third challenge is that Mizar does not have a ``small trusted kernel'' in any reasonable sense. The main components that perform trusted activities are the ``analyzer'' and the ``checker'', and most of the code of the verifier is contained in these modules. Moreover, the code is written in Pascal and Polish (two languages that the author is not very good at), and is 50 years old, meaning that there is a huge amount of technical debt in the code.\footnote{Technically the current incarnation of Mizar (``PC Mizar'') only dates to 1986. See \cite{matuszewski2005mizar} for the early history of Mizar.} To get a good sense of what everything is doing and why it is correct (or not), we estimated the best option would simply be to rewrite it and closely follow what the original code is doing. Rust was chosen as the implementation language because it is well suited to writing console applications with a focus on correctness, performance, and refactoring, all of which were important for something that could play the role of \texttt{verifier}.

We are not the first to present aspects of Mizar from an outsider's perspective. J. Harrison's Mizar mode in HOL Light \cite{harrison-mizar} is a simulation of large parts of the Mizar frontend, with the MESON prover used in place of the Mizar checker. J. Urban's MPTP \cite{urban2006mptp} project exports the statements of checker goals and sends them to ATPs. And most recently, Kaliszyk and P{\k{a}}k presented a model of Mizar as encoded in Isabelle \cite{kaliszyk2019semantics}, but stop short of a fully automatic proof export mechanism.

One may reasonably ask why the MPTP export is not already everything you could wish regarding proof export. There are two problems: (1) Because it is using external ATPs rather than Mizar's own code, it does not and cannot reasonably be expected to achieve 100\% success on the MML, and even when it works the Mizar checker often outperforms the external ATP, because the MML is significantly overfit to the Mizar checker. (2) Because it instruments the ``checker inference'' level, we verify ``\texttt{by}'' steps but not the skeleton steps that glue them together, and some significant type reasoning happens in the analyzer.

In this paper, we will try to give an implementer's perspective of Mizar: how it works under the hood, and some things we discovered while trying to replicate the functionality. The reader should be aware that this paper is focused on internal details and implementation correctness of the Mizar system, not the abstract theory of Mizar (that is, Tarski--Grothendieck set theory and the soft type system). See \cite{grabowski2010mizar,kaliszyk2019semantics} for information on these aspects of the language.

(This is an expanded version of a paper presented at ITP 2023. This paper includes major revisions of \autoref{sec:user-internals} and \autoref{sec:results}, and new sections \ref{sec:nameres} and \ref{sec:parser}.)

\section{Mizar internals: A (determined) user's perspective}\label{sec:user-internals}

When one runs the \texttt{verifier} on a Mizar file, say \texttt{tarski.miz}, one will see something like this:
\begin{verbatim}
Verifier based on More Strict Mizar Processor, Mizar Ver. 8.1.11 (Linux/FPC)
Copyright (c) 1990-2022 Association of Mizar Users
Processing: mizshare/mml/tarski.miz

Parser   [ 138]   0:00
MSM      [ 132]   0:00
Analyzer [ 137]   0:00
Checker  [ 137]   0:00
Time of mizaring: 0:00
\end{verbatim}
This refers to modules named \texttt{Parser}, \texttt{MSM}, \texttt{Analyzer}, \texttt{Checker}, so one can reasonably guess that the code is split up into multiple passes and these are the names of the modules.
\begin{itemize}
    \item \texttt{Parser} reads the text and converts it into an abstract syntax tree.
    \item \texttt{MSM} stands for ``More Strict Mizar'', which is a slightly more elaborated version of the syntax tree. Notably, this pass does (some kinds of) name resolution, as well as resolving and auto-binding reserved variables.
    \item \texttt{Analyzer} is responsible for processing the large scale logical structure of a Mizar document. It resolves symbol overloading and checks types, as well as tracking the ``thesis'' as it is transformed by Mizar proof steps.
    \item \texttt{Checker} is called by the analyzer whenever there are atomic proof obligations. It is sometimes referred to as the ``\lstinline|by|'' proof automation because it is called whenever there is a step proved using the \lstinline|by| keyword.
\end{itemize}

Both the analyzer and the checker are soundness-critical components. Clearly the checker is important because it actually checks proofs -- a bug here will mean that an incorrect proof is accepted -- but the analyzer is also important because it determines the proof obligations that will be sent to the checker. If the analyzer sends the wrong assertion to the checker or simply doesn't call the checker at all, then it might assert a theorem that hasn't been properly checked.

\begin{figure}[htb]
  \begin{center}
  \begin{tikzpicture}[auto,every node/.style={fill=white}]
    \node at (0,-2) [fill=gray!30, shape=rounded rectangle] (accom) {accom};
    \node at (2.5,-1.5) [fill=gray!30, shape=rounded rectangle] (parser) {parser\rlap{\phantom{l}}};
    \node at (2.5,-4) [fill=gray!30, shape=rounded rectangle] (msm) {MSM};
    \node at (2.5,-6) [fill=gray!30, shape=rounded rectangle] (analyzer) {analyzer};
    \node at (2.5,-8) [fill=gray!30, shape=rounded rectangle] (checker) {checker\rlap{\phantom{pl}}};
    \node at (-3,-8) [fill=gray!30, shape=rounded rectangle] (exporter) {exporter\rlap{\phantom{pl}}};

    \node at (0,0) (global) [draw, align=center, anchor=north, label={north:{\emph{global data}}}] {%
      {\texttt{mml.ini}}\\[-2pt]
      {\texttt{mml.vct}}\\[-2pt]
      {\texttt{mizar.dct}}
    };
    \draw[->] (global) -- (accom);

    \node at (-6,-0.5) [draw, anchor=north] (dre) {\color{red}\emph{dep}\texttt{.dre}};
    \draw[->] (dre) .. controls +(4,0) and +(-.3,.4) .. (accom.150);

    \node (g1) [draw, align=center, below=.2 of dre] {%
      {\color{red}\emph{dep}\texttt{.dcl}}\\[-2pt]
      {\color{red}\emph{dep}\texttt{.dco}}\\[-2pt]
      {\color{red}\emph{dep}\texttt{.def}}\\[-2pt]
      {\color{red}\emph{dep}\texttt{.did}}\\[-2pt]
      {\color{red}\emph{dep}\texttt{.dno}}\\[-2pt]
      {\color{red}\emph{dep}\texttt{.dpr}}\\[-2pt]
      {\color{red}\emph{dep}\texttt{.drd}}\\[-2pt]
      {\color{red}\emph{dep}\texttt{.sch}}\\[-2pt]
      {\color{red}\emph{dep}\texttt{.the}}
    };
    \coordinate (exporter_out) at (-5,-5);
    \draw[->] (exporter) .. controls +(-2,0) and +(.5,0) .. (g1.295);
    \draw[->] (g1.65) .. controls +(1,0) and +(-.6,.3) .. (accom.165);

    \node at (2.5,-0.1) [draw, anchor=north] (miz) {\texttt{.miz}};
    \draw[->] (miz) -- (accom);
    \draw[->] (miz) -- (parser);

    \node at (-4,-2.5) (g2) [draw, align=center, anchor=north] {%
      {\color{red}\texttt{.aco}}\\[-2pt]
      {\color{red}\texttt{.ano}}\\[-2pt]
      {\texttt{.dct}}\\[-2pt]
      {\color{red}\texttt{.dcx}}\\[-2pt]
      {\color{red}\texttt{.evl}}\\[-2pt]
      {\texttt{.fil}}\\[-2pt]
      {\color{red}\texttt{.frm}}\\[-2pt]
      {\texttt{.nol}}
    };
    \draw[->] (accom.180) .. controls +(-3,0) and +(0,.5) .. (g2.90);

    \node at (2.5,-7) [draw] (xml) {\color{red}\texttt{.xml}};
    \draw[->] (analyzer) -- (xml);
    \draw[->] (xml) -- (checker);
    \draw[->] (xml) -- (exporter);

    \node at (2.5,-9) [scale=1.25] (check) {\checkmark};
    \draw[->] (checker) -- (check);

    \node at (-1,-2.5) (g4) [draw, align=center, anchor=north] {%
      {\color{red}\texttt{.dfe}}\\[-2pt]
      {\color{red}\texttt{.dfx}}\\[-2pt]
      {\color{red}\texttt{.eid}}\\[-2pt]
      {\color{red}\texttt{.erd}}\\[-2pt]
      {\color{red}\texttt{.esh}}\\[-2pt]
      {\color{red}\texttt{.eth}}
    };
    \coordinate (checker_in) at (-1,-6);
    \draw[->] (accom.180) .. controls +(-.3,0) and +(0,.3) .. (g4.90);
    \draw [white,line width=4pt,opacity=1] (checker_in) .. controls +(0,-1) and +(0,0) .. (checker.170);
    \draw[->] (g4) -- (checker_in) .. controls +(0,-1) and +(0,0) .. (checker.170);

    \node at (0,-2.5) [draw, anchor=north] (epr) {\color{red}\texttt{.epr}};
    \draw[->] (accom) -- (epr);
    \draw (epr) .. controls +(0,-2.6) and +(0,.9) .. (checker_in);
    \draw[->] (epr.270) .. controls +(0,-.6) and +(0,0) .. (analyzer.135);

    \node at (1,-2.5) (g3) [draw, align=center, anchor=north] {%
      {\color{red}\texttt{.dfs}}\\[-2pt]
      {\color{red}\texttt{.eno}}
    };
    \draw[->] (accom.0) .. controls +(.3,0) and +(0,.3) .. (g3.90);
    \draw[->] (g3) .. controls +(0,-1.8) and +(-.25,0.3) .. (analyzer.135);

    \node at (-3,-2.5) (g5) [draw, align=center, anchor=north] {%
      {\texttt{.cho}}\\[-2pt]
      {\texttt{.prf}}\\[-2pt]
      {\texttt{.sgl}}\\[-2pt]
      {\color{red}\texttt{.vcl}}
    };
    \coordinate (exporter_in) at (-3,-6);
    \draw[->] (accom.180) .. controls +(-2,0) and +(0,.5) .. (g5.90);
    \draw[->] (g5) -- (exporter_in) -- (exporter);

    \node at (4.35,-2.5) (g6) [draw, align=center, anchor=north] {%
      {\texttt{.err}}\\[-2pt]
      {\texttt{.log}}
    };
    \draw[->] (accom) -- (1,-2) .. controls +(2,0) and +(0,0) .. (g6.160);
    \draw[->] (parser) -- (g6);

    \node at (-2,-2.5) (g7) [draw, align=center, anchor=north] {%
      {\color{red}\texttt{.atr}}\\[-2pt]
      {\color{red}\texttt{.ecl}}\\[-2pt]
      {\texttt{.ere}}
    };
    \draw[->] (accom.180) .. controls +(-1,0) and +(0,.5) .. (g7.90);
    \draw (g7) .. controls +(0,-2) and +(0,1) .. (checker_in);
    \draw (g7) .. controls +(0,-2) and +(0,1) .. (exporter_in);

    \node at (0,-9) [draw] (bex) {\color{red}\texttt{.bex}};
    \draw[->] (checker.south) .. controls +(0,-.2) and +(0,.5) .. (bex.north);

    \node at (-1,-9) [draw] (fex) {\color{red}\texttt{.fex}};
    \draw[->] (checker.south) .. controls +(0,-.2) and +(0,.5) .. (fex.north);
    \draw[->] (exporter) -- (fex);

    \node at (-6,-5.5) [draw] (dfr) {\color{red}\texttt{.dfr}};
    \draw[->] (exporter) .. controls +(-1.9,0) and +(.9,0) .. (dfr.east);

    \node at (4.35,-1.5) [draw] (idx) {\color{red}\texttt{.idx}};
    \draw[->] (parser) edge[bend left=10] (idx);
    \draw[->] (idx) edge[bend left=10] (parser);

    \node at (2.5,-2.5) [draw, align=center, anchor=north] (wsx) {%
      {\color{red}\texttt{.frx}}\\
      {\color{red}\texttt{.wsx}}
    };
    \draw [white,line width=4pt,opacity=1] (parser) -- (wsx);
    \draw[->] (parser) -- (wsx);
    \draw[->] (wsx) -- (msm);

    \node at (2.5,-5) [draw] (par) {\texttt{.par}};
    \draw[->] (msm) -- (par);
    \draw[->] (par) -- (analyzer);

    \node at (0.5,-5) [draw] (ref) {\texttt{.ref}};
    \draw [white,line width=4pt,opacity=1] (msm) -- (ref);
    \draw[->] (msm) -- (ref);
    \draw (ref.180) .. controls +(0,0) and +(0,0.5) .. (checker_in);

    \node at (4.35,-4) [draw] (msx) {\color{red}\texttt{.msx}};
    \draw[->] (msm) edge[bend left=10] (msx);
    \draw[->] (msx) edge[bend left=10] (msm);

    \node [draw, dotted, fill=none, fit=(dre) (g1), inner sep=6pt,label={north:{\emph{in \texttt{prel/}}}}] {};

    \node [draw, dotted, fill=none, fit=(miz) (parser) (wsx) (msm) (par) (analyzer) (xml) (checker) (check), inner sep=3pt, label={north:{\emph{main verifier}}}] {};

    \node [draw, dotted, fill=none, fit=(fex) (bex), inner sep=4.5pt, label={west:{\emph{MPTP proof export}}}] {};
  \end{tikzpicture}
  \end{center}
  \caption{Data flow between Mizar components. Files with only an extension like \texttt{.miz} correspond to files named after the input \emph{art}\texttt{.miz} file. Files in {\color{red}red} are in XML format. For example, {\color{red}\emph{art}\texttt{.atr}} is created by the accommodator and read by the checker and the exporter. Files like {\color{red}\emph{dep}\texttt{.dcl}} are produced for the current article by the exporter and aggregated by the accommodator in articles importing that dependency.}
  \label{fig:flow}
\end{figure}

Another visible effect of running \texttt{verifier} on a \texttt{.miz} file is that a huge number of other files are generated with bespoke extensions, which reveals another implementation quirk of the Mizar system, which is that the four components above are very loosely coupled and communicate only via the file system (see \autoref{fig:flow}). The parser reads the \texttt{.miz} file and produces a \texttt{.wsx} file, which is read by \texttt{MSM} to produce \texttt{.msx}, which is read by the \texttt{transfer2analyzer} module (not mentioned above) to produce \texttt{.par}, which is read by \texttt{Analyzer} to produce \texttt{.xml}, which is read by \texttt{Checker} and verified.

The general data flow of a Mizar verification looks like this:
\begin{itemize}
  \item The first tool to process the article is the ``accommodator,'' \texttt{accom}. This is the only tool which reads data from articles other than the current one; it is responsible for aggregating data from imports and collecting them for the current article. The many files produced in this pass are consumed as needed by later stages in the pipeline, to provide information about the extant declarations of various kinds in the environment.

  \item The next pass is the parser, which parses the body of the \textit{art}\texttt{.miz} article and produces a {\color{red}\textit{art}\texttt{.wsx}} file (Weakly Strict Mizar) containing the AST, along with {\color{red}\textit{art}\texttt{.frx}} containing the formats declared in the article.

  \item The MSM pass reads {\color{red}\textit{art}\texttt{.wsx}} and writes {\color{red}\textit{art}\texttt{.msx}}, performing name resolution and filling in the types of \texttt{reservation} variables in statements.

  \item The \texttt{transfer2analyzer} pass appears to be for backward compatibility reasons, as it reads {\color{red}\textit{art}\texttt{.msx}} and translates it to \textit{art}\texttt{.par}, which is the same thing but in a less extensible format. It also resolves format references.

  \item The analyzer pass reads \textit{art}\texttt{.par} and writes {\color{red}\textit{art}\texttt{.xml}}, which is the article AST again but using fully elaborated and typechecked terms, and with all the statements of checker subgoals explicitly annotated.

  \item The checker is the last pass. It reads \textit{art}\texttt{.xml} and verifies the theorems.
\end{itemize}

Thanks to the efforts of J. Urban in 2004 \cite{xmlizing}, most of these internal files between the components are in XML format, highlighted in {\color{red}red}. Of the non-XML files, most of them have a simple format, with the notable exception of \emph{art}\texttt{.miz} (see \autoref{sec:parser}) and \emph{art}\texttt{.par} (which we managed to avoid reading or writing by using \emph{art}\texttt{.msx} instead).

(\texttt{mizar-rs} has a significantly different architecture. One can find identifiable components corresponding to ``accom'', ``parser'', ``MSM + analyzer'', ``checker'', and ``exporter'', and for each PC Mizar component one can either enable it in \texttt{mizar-rs}, in which case it will read the input files of that component and perform the work, or disable it in which case it will read the output of that component. This allows mixing and matching between the two implementations, for example running \texttt{accom} to generate all the intermediate files, and then running only the ``checker'' component of \texttt{mizar-rs}, which will expect e.g. \emph{art}\texttt{.eth} to exist. If all components are enabled, only the \emph{dep}\texttt{.dre} files, the global files \texttt{mml.ini}, \texttt{mml.vct}, \texttt{mizar.dct}, and the article source files \emph{art}\texttt{.miz} are used, all of which are human-written input files.)

If one opens one of these files, one is presented with the next major challenge, which is that terms are pervasively indexed and hence reading the expressions can be quite difficult. Moreover, there are many \emph{distinct} index classes which are differentiated only by the context in which they appear, so without knowing how the program processes the indices or what array is being accessed it is hard as an outsider to follow the references. In the Rust implementation this issue is addressed by using ``newtypes'' to wrap each integer to help distinguish different indexing sets. There are currently 36 of these newtypes defined: for example there are numbers for functors, selectors, predicates, attributes, formats, notations, functor symbols (not the same!), left bracket symbols, reserved identifiers, etc.

After parsing these, one ends up with an expression such as\\ $\forall \_:M_1,\ R_4(K_1(B_1,K_2(N_2)),N_1)$, which means something like ``for all $x$ of the first type, fourth relation holds of the first function applied to $x$ and the second function applied to $2$, and $1$''. For many purposes, this is sufficient for debugging, but one tends to go cross-eyed staring at these expressions for too long. Ideally we would be able to reverse this indexification to obtain the much more readable expression $\forall x:{\bf Nat},x\cdot (-2) \le 1$. However, the code to do this does not exist anywhere in Mizar, because the Mizar checker never prints expressions. The only output of the checker is a list of (\texttt{line}, \texttt{col}, \texttt{err\_code}) triples, which are conventionally postprocessed by the \texttt{errflag} tool to insert markers like the following in the text:

\begin{lstlisting}
for x being Nat holds x = x
proof
  let x be Nat;
end;
::>,70
::>
::> 70: Something remains to be proved
\end{lstlisting}
This error message is pointing at the \lstinline|end| keyword, but notably it does not say \emph{what} remains to be proved, here $x=x$. A debug build of Mizar will actually print out expressions to the \textit{art}\texttt{.inf} file, but they are similar to the $R_4(K_1(\dots))$ style.

Luckily, this issue has been addressed outside the main Mizar codebase: J. Urban's HTMLizer is an XSLT stylesheet which can transform the XML intermediate files into fully marked up reconstructed Mizar documents, and which we adapted to design the formatter.

With the present version of the formatter, and with appropriate debugging enabled, an input like the following:
\begin{lstlisting}
for x,y being Nat holds x = 1 & y = 2 implies x + y = 3;
\end{lstlisting}
yields this debugging trace:
\begin{verbatim}
input: ∃ b0: natural set, b1: natural set st
  (b0 = 1) ∧ (b1 = 2) ∧ ¬((b0 + b1) = 3)
refuting 0 @ TEST:28:33:
  ∃ b0: natural set, b1: natural set st
    (b0 = 1) ∧
      (b0 c= 1) ∧
      (1 c= b0) ∧
      (b1 = 2) ∧
      (b1 c= 2) ∧
      (2 c= b1) ∧
      ¬((b0 + b1) = 3) ∧
      (((b0 + b1) c= 3) → ¬(3 c= (b0 + b1)))
\end{verbatim}
In addition to showing some of the formatting and indentation behavior of the reconstructed expression, this also reveals some aspects of the checker, like how the goal theorem has been negated and the definitional theorem $x=y\leftrightarrow x\subseteq y\land y\subseteq x$ has been eagerly applied during preprocessing.

One other feature that is demonstrated here is ``negation desugaring'', which requires some more explanation. Internally, Mizar represents all expressions using only $\neg$, $\forall$ and n-ary $\land$. So $P\to Q$ is mere syntax for $\neg(P\land \neg Q)$ and $\exists x, P(x)$ is actually $\neg\forall x, \neg P(x)$. (Even $P\leftrightarrow Q$ is desugared, to $\neg(P\land \neg Q)\land \neg(Q\land \neg P)$, so too much recursive use of $\leftrightarrow$ can cause a blowup in formula size.) This normalization ensures that different spellings of the same formula are not distinguished, for example if the goal is $P\lor Q$ then one may prove it by \lstinline|assume not P; thus Q;|. Double negations are also cancelled eagerly. In the formatter, we try to recover a natural-looking form for the expression by pushing negations to the leaves of an expression, and also writing $A\land B\to C\lor D$ if after pushing negations we get a disjunction such that the first few disjuncts have an explicit negation. (Mizar actually carries some annotations on formulas to help reconstruct whether the user wrote \lstinline|not A or B| or \lstinline|A implies B|, but we chose not to use this information as it is often not available for expressions deep in the checker so we wanted a heuristic that works well in the absence of annotation.)

\section{The checker}\label{sec:checker}
Mizar is broadly based on first order logic, with the non-logical axioms of Tarski--Grothendieck set theory, with a type system layered on top. While types can depend on terms, so one may call it a dependent type theory, this is not a type system in the sense of Martin-L\"{o}f type theory: types are essentially just predicates over terms in an untyped base logic, and the language allows typing assertions to be treated as predicates.

\subsection{Core syntax}\label{sec:syntax}

The core syntax of Mizar uses the following grammar:
$$a,b,c,d,e,F,G,H,K,U,V,P,R,S,T_M,T_G::=\mbox{ident}$$
\[
\begin{minipage}{.55\linewidth}
  \centering
  $\begin{array}{r@{}l@{\quad}l}
    t \mathrel{::=}{} & a,b,c,d,e & \mbox{variable} \\
    \mid{} & n & \mbox{numeral} \\
    \mid{} & \{F,G,K,U\}(\overrightarrow{t_i}) & \mbox{function application} \\
    \mid{} & H(\overrightarrow{t_i}):=t' & \mbox{local function app.} \\
    \mid{} & \operatorname{\bf the}\,\tau & \mbox{choice} \\
    \mid{} & \{t\mid \overrightarrow{b_i:\tau_i} \mid \varphi\}& \mbox{Fraenkel} \\[1ex]
    \tau \mathrel{::=}{} & \overrightarrow{\chi_i}\;T_M(\overrightarrow{t_i}) & \mbox{mode} \\
    \mid{} & \overrightarrow{\chi_i}\;T_G(\overrightarrow{t_i}) & \mbox{struct type} \\[1ex]
    \chi \mathrel{::=}{} & \pm V(\overrightarrow{t_i}) & \mbox{attribute} \\
  \end{array}$
\end{minipage}%
\begin{minipage}{.45\linewidth}
  \centering
  $\begin{array}{r@{}l@{\quad}l}
    \varphi \mathrel{::=}{} & \top & \mbox{true} \\
    \mid{} & \neg \varphi & \mbox{negation} \\
    \mid{} & \bigwedge \overrightarrow{\varphi_i} & \mbox{conjunction} \\
    \mid{} & \bigwedge_{b=t}^{t'} \varphi & \mbox{flex conjunction} \\[\jot]
    \mid{} & \forall\,b:\tau.\; \varphi & \mbox{for all} \\
    \mid{} & \{P,R\}(\overrightarrow{t_i}) & \mbox{predicate} \\
    \mid{} & S(\overrightarrow{t_i}):=t' & \mbox{local pred.} \\
    \mid{} & t\is\chi & \mbox{attribute} \\
    \mid{} & t:\tau & \mbox{qualification} \\
  \end{array}$
\end{minipage}
\]

Broadly speaking, terms are first order, meaning applications of function symbols to variables. The exception is the Fraenkel operator $\{t\mid \overrightarrow{b_i:\tau_i} \mid \varphi\}$, which should be read as ``the set of $t(\overrightarrow{b_i})$ where the $\overrightarrow{b_i:\tau_i}$ are such that $\varphi(\overrightarrow{b_i})$'', where the variables $\overrightarrow{b_i:\tau_i}$ are quantified in both $t$ and $\varphi$.

All terms $t$ have a type $\tau$, and there is a robust subtyping system -- most terms have \emph{many} types simultaneously. All types are required to be nonempty, which is what justifies the $\operatorname{\bf the}\,\tau$ constructor for indefinite description. Types are composed of a collection of attributes (a.k.a. clusters) $\chi$ applied to a base type $\{T_M,T_G\}(\overrightarrow{b_i})$. Regular types are called ``modes''; new modes can be defined by carving out a subset of an existing mode, and modes need not define a set (in particular, \lstinline|object| is a primitive mode which is the supertype of everything, and does not constitute a set per Russell's paradox). Structure types are roughly modeled after partial functions from some set of ``tags'', but they are introduced axiomatically, similarly to how structure types are treated in a dependent type theory such as Coq or Lean.

\begin{remark}
Although one can define a mode for any ZFC set or class, modes and sets are not interchangeable in Mizar because they lie in different syntactic classes. Modes are types, which go in the type argument of quantifiers such as the ${\bf Nat}$ in $\forall x:{\bf Nat}.\ x\ge 0$, while sets are objects, which can be passed to functions, like $\left|\texttt{NAT}\right|=\aleph_0$. In the MML, \texttt{NAT} is the set of natural numbers, while \texttt{Nat} and \texttt{Element of NAT} are modes which describe the type of natural numbers (and in general \texttt{Element of A} can be used to treat a set as a type).
\end{remark}

Attributes, also known as ``adjectives'' or ``clusters'', are modifiers on types, which may be described as intersection typing in modern terminology, although attributes do not stand alone as types. For example, ``\texttt{x is non empty finite set}'' means $\neg(x\is\texttt{empty})\wedge (x\is\texttt{finite})\wedge (x:\texttt{set})$. The Mizar system treats the collection of attributes on a type as an unordered list for equality comparisons.

Formulas are composed mainly from negation, conjunction and for-all, but there are some extra formula constructors that deserve attention:
\begin{itemize}
  \item Not represented in the grammar is that $\neg\neg\varphi$ is identified with $\varphi$, and internally there is an invariant that $\neg\neg$ never appears.
  \item Similarly, conjunctions are always flattened, so $P\land(Q\land R)\land S$ becomes $P\land Q\land R\land S$.
  \item Unlike most type theories, typing assertions are reified into an actual formula, so it is possible to say $4/2:\N$ and $\neg (-1:\N)$.
  \item There is also a predicate for having an attribute, and $(t:\chi\;\tau)\leftrightarrow (t\is\chi)\land (t:\tau)$ is provable.
  \item Flex-conjunction is written with syntax such as \verb|P[1] & ... & P[n]| (with literal ``\verb|...|''), and Mizar knows that this expression is equivalent both to $\forall x:\N.\ 1\le x\le n\to P(x)$ as well as to an explicit conjunction (when $n$ is a numeral). For example \verb|P[1] & ... & P[3]| will be expanded to \verb|P[1] & P[2] & P[3]| in the checker.
\end{itemize}

The different letters for functions, predicates, and variables correspond to the different roles that these can play, although in most situations they are treated similarly.
\begin{itemize}
  \item ``Locus variables'' ($a$) are used in function declarations to represent the parameters of a function, mode, cluster, etc. For example a local function might be declared as $K(\overrightarrow{a_i}):\tau:=t$ where $\tau$ and $t$ are allowed to depend on the $\overrightarrow{a_i}$, and then when typing it we would have that $K(\overrightarrow{t_i})$ has type $\tau[\overrightarrow{a_i\mapsto t_i}]$.
  \item ``Bound variables'' ($b$) are de Bruijn \emph{levels}, used to represent variables in all binding syntaxes. So for example, $\forall x:\N.\,\forall y:\N.\,x\le x+y$ would be represented as\\ $\forall \underline{\ }:\N.\,\forall \underline{\ }:\N.\,b_0\le b_0+b_1$. NB: there are two conventions for locally nameless variables, called ``de Bruijn indices'' (variables are numbered from the inside out) and ``de Bruijn levels'' (variables are numbered from the outside in), and Mizar's choice of convention is the less common one.
  \item ``Constants'' ($c$) are variables that have been introduced by a \lstinline|consider|, \lstinline|given|, or \lstinline|take| declaration, as well as Skolem constants introduced in the checker when there are existentially quantified assumptions. Constants may or may not be defined to equal some term (for example \lstinline|take x = t| will introduce a variable $x$ equal to $t$) and the checker will use this in the equalizer if available.
  \item ``Inference constants'' ($d$) are essentially checker-discovered abbreviations for terms. This is used to allow for subterm sharing, as terms are otherwise completely unshared.
  \item ``Equivalence classes'' ($e$) are used in the equalizer to represent equivalence classes of terms up to provable equality. So for example if the equalizer sees a term $x+y$ it will introduce an equivalence class $e_1$ for it and keep track of all the things known to be in this class, say $e_1=x+y=y+x=e_1+0$.
  \item Functors ($K$) and predicates ($R$) are the simplest kind of definition, they correspond to function and relation symbols in traditional FOL.
  \item Scheme functors ($F$) and scheme predicates ($P$) are the higher-order analogue of constants. They are only valid inside scheme definitions, which declare these at the start and then use them in the statement and proof.
  \item Local functors ($H$) and local predicates ($S$) are function declarations within a local scope, declared with the \lstinline|deffunc| or \lstinline|defpred| keywords. Internally every predicate carries the result of substituting its arguments, so if we declare $\operatorname{\bf deffunc}H_1(x,y):=x+y$ then an expression like $2 \cdot H_1(x-1,y)$ is really stored as $2 \cdot (H_1(x-1,y):=x-1+y)$. Some parts of the system treat such an expression like a function application, while others treat it like an abbreviation for a term.
  \item An aggregate functor ($G$) is Mizar terminology for a structure constructor $\{\mathrm{foo}:=x,\mathrm{bar}:=y\}:\mathrm{MyStruct}$, and the converse operator is a selector ($U$), which is the projection $t.\mathrm{foo}$ for a field. (The Mizar spelling for these operators is \lstinline|MyStruct(#|\ \,\lstinline|x, y #)| for the constructor and \lstinline|the foo of t| for the projection.)
\end{itemize}

\subsection{Structure of the checker}\label{sec:structure}

The checker is called whenever there is a proof such as ``\lstinline|2 + 2 = 4;|'' or ``\lstinline|x <= y by A1,Thm2;|'': basically any time there is a \lstinline|by| in the proof text, as well as when propositions just end in a semicolon, this being the nullary case of \lstinline|by|. It consists of three major components, although there is another piece that is used even before the checker is properly called on a theorem statement:

\begin{enumerate}
  \setlength{\itemsep}{1ex}
  \item[0.] \textbf{Attribute inference} (``rounding-up'') is used on every type $\overrightarrow{\chi}^-\tau$ to prove that it is equivalent to $\overrightarrow{\chi}^+\tau$ where $\overrightarrow{\chi}^+$ is a superset of $\overrightarrow{\chi}^-$. Internally we actually keep \emph{both} versions of the attributes (the ``lower cluster'' $\overrightarrow{\chi}^-$ being the one provided by the user and the ``upper cluster'' $\overrightarrow{\chi}^+$ being what we can infer by applying all inference rules), because it is useful to be able to prove that two clusters are equal if $\overrightarrow{\chi}_1^-\subseteq \overrightarrow{\chi}_2^+$ and $\overrightarrow{\chi}_2^-\subseteq \overrightarrow{\chi}_1^+$.

  \item The \textbf{pre-checker} takes the list of assumptions, together with the negated conjecture, and performs a number of normalizations on them, skolemizing existentials, removing vacuous quantifiers, and expanding some definitions. Then it converts the whole formula into \emph{disjunctive} normal form (DNF) and tries to refute each clause.

  \item The \textbf{equalizer} does most of the ``theory reasoning''. It replaces each term in the clause with an equivalence class, adding equalities for inference constants and defined constants, as well as registered equalities and \lstinline|symmetry| declarations, as well as any equalities in the provided clause. It also evaluates the numeric value (for $2+2=4$ proofs) and polynomial value (for $(x+1)^2=x^2+2x+1$ proofs) of equality classes and uses them to union classes together. These classes also have many \emph{types} since many terms are going into the classes, and all of the attributes of these types are mixed together into a ``supercluster'' and rounded-up some more, potentially leading to a contradiction. For example if we know that $x=y$ and $x$ is positive and $y$ is negative then the supercluster for $\{x,y\}$ becomes \lstinline|positive negative| which is inconsistent.

  \item The last step is the \textbf{unifier}, which handles instantiation of quantifiers. This is relatively simplistic and non-recursive: for each assumption $\forall \vec{b}.\;P(\vec{b})$ it will instantiate $P(\vec{v})$ (where $\vec{v}$ are metavariables) and then construct the possible assignments (as a DNF) to the variables that would make $P(\vec{v})$ inconsistent with some other assumption. It tries this for each forall individually, and if it fails, it tries taking forall assumptions in pairs and unifying them. If that still fails then it gives up -- it does not attempt a complete proof method.
\end{enumerate}

\noindent Finally, there is one more component which is largely separate from the ``\lstinline|by|'' automation:

\begin{enumerate}
  \item[4.] The \textbf{schematizer} is the automation that is called for justifications starting with ``\lstinline|from|''. These are scheme instantiations. In this case it is very explicitly given the list of hypotheses in the right order, so the only thing it needs to do is to determine an assignment of the scheme variables ($F,P$) to regular functors and predicates ($K,R$) or local functors and predicates ($H,S$). Users are often required to introduce local predicates in order to apply a scheme. Nullary scheme functors (a.k.a constant symbols) can be unified with arbitrary terms, however.
\end{enumerate}

Although we cannot go into full detail on the algorithms here, in the following sections we will go into some of the highlights, with an emphasis on what it takes to audit the code for logical soundness, through the lens of root-causing some soundness bugs.

\subsection{Requirements}\label{sec:requirements}

One aspect of the Mizar system that is of particular interest is the concept of ``requirements'', which are definitions that the checker has direct knowledge of. For example, the grammar given above does not make any special reference to equality: it is simply one of the possible relation symbols $R_i(x,y)$, and the relation number for equality can vary from one article to the next depending on how the accommodator decides to order the imported relation symbols. Nevertheless, the checker clearly needs to reason about equality to construct equality equivalence classes.

To resolve this, there is a fixed list of ``built-in'' notions, and one of the files produced by the accommodator (\textit{art}\texttt{.ere}) specifies what the relation/functor/mode/etc. number of each requirement is. Importantly, if a constructor is identified in this way as a requirement, this not only allows the checker to recognize and produce expressions like $x=y$, it is \emph{also an assertion that this relation behaves as expected}. If a requirement is given a weird definition, for example if we were to open \texttt{xcmplx\_0.miz} and change the definition of $x+y$ to mean subtraction instead, we would be able to prove false in a downstream theorem which enables the requirement for $+$, because we would still be able to prove $2+2=4$ by evaluation.

At the surface syntax level, requirements are enabled in groups, using the \lstinline|requirements| directive in the import section. The requirements are, in rough dependency order:

\begin{itemize}
  \setlength{\itemsep}{1ex}
  \item \texttt{HIDDEN} (introduced after the \texttt{HIDDEN} article) is a requirement that is automatically enabled for every Mizar file. It introduces the modes \lstinline|object| and \lstinline|set|, as well as $x=y$ and $x\in y$. (The article \texttt{HIDDEN} itself is somewhat magical, and cannot be processed normally because every file takes an implicit dependency on \texttt{HIDDEN}.)
  \item \texttt{BOOLE} (introduced after \texttt{XBOOLE\_0}) introduces the adjective \lstinline|x is empty|, as well as set operators $\emptyset$, $A\cup B$, $A\cap B$, $A\setminus B$, $A\oplus B$, and $A\operatorname{meets}B$ (i.e. $A\cap B\ne \emptyset$). These operators have a few extra properties such as $A\cup \emptyset=A$ that are used in the equalizer.
  \item \texttt{SUBSET} (introduced after \texttt{SUBSET\_1}) introduces the mode \lstinline|Element of A|, along with $\mathcal{P}(A)$, $A\subseteq B$, and the mode \lstinline|Subset of A|. The checker knows about how these notions relate to each other, for example if $x\in A$ then $x:\mathtt{Element}(A)$. (Because types have to be nonempty, $x:\mathtt{Element}(A)$ is actually equivalent to $x\in A\lor (A=\emptyset\land x=\emptyset)$. So the reverse implication used by the checker is $(x:\mathtt{Element}(A))\land \neg (A\is\mathtt{empty})\to x\in A$.)
  \item \texttt{NUMERALS} (introduced after \texttt{ORDINAL1}) introduces $\operatorname{succ}(x)$, \lstinline|x is natural|, the set $\N$ (spelled \lstinline|NAT| or \lstinline|omega|), $0$, and \lstinline|x is zero|. Note that $0$ is not considered a numeral in the sense of section \ref{sec:syntax}, it is a functor symbol $K_i()$. Numbers other than $0$ can be written even before the \texttt{ORDINAL1} article, but they are uninterpreted sets; after this requirement is added the system will give numerals like $37$ the type \lstinline|Element of NAT| instead of \lstinline|set|.
  \item \texttt{REAL} (introduced after \texttt{XXREAL\_0}) introduces $x\le y$, \lstinline|x is positive|, \lstinline|x is negative|, along with some basic implications of these notions.
  \begin{itemize}
    \item \texttt{NUMERALS + REAL} enables the use of flex-conjunctions $\bigwedge_{i=a}^b\varphi(i)$, since these expand to the expression $\forall i:\N.\ a\le i\le b\to \varphi(i)$ which requires $\N$ and $\le$ to write down.
  \end{itemize}
  \item \texttt{ARITHM} (introduced after \texttt{XCMPLX\_0}) introduces algebraic operators on the complexes: $x+y$, $x\cdot y$, $-x$, $x^{-1}$, $x-y$, $x/y$, $i$, and \lstinline|x is complex|. This also enables the ability to do complex rational arithmetic on numerals, as well as polynomial normalization.
\end{itemize}

\subsubsection{Soundness considerations of the requirements}\label{sec:req-soundness}

There is a slight mismatch between what the user has to provide in order to enable a requirement and what the checker gets to assume when a requirement is enabled which causes a challenge for proof export or other external soundness verification. As the list above might indicate, generally checker modules corresponding to a requirement are enabled as soon as all of the \emph{constructors} involved in stating them are available; for example we can see this with the \texttt{NUMERALS + REAL} prerequisite for flex-conjunctions. (More precisely, flex-conjunctions are enabled exactly when $\N$ and $\le$ become available.) However, the checker needs more than that to justify the manipulations it does with them.

For example the checker exploits the fact that $\bigwedge_{i=1}^3\varphi(i)$ is equivalent both to $\forall i:\N.\ 1\le i\le 3\to \varphi(i)$ and to $\varphi(1)\land \varphi(2)\land \varphi(3)$, and so this amounts to an assertion that $a\le i\le b\leftrightarrow i=a\lor i=a+1\lor\dots\lor i=b$ is provable when $a$ and $b$ are numerals such that $a\le b$. The reverse implication follows from numerical evaluation, and the forward implication is a metatheorem that can be proven by induction, assuming the existence of lemma L saying $a\le i\to i=a\lor \mathrm{succ}(a)\le i$:

\begin{theorem}
  If $a$ and $n$ are numerals and $i:\N$, then
  \vspace{-3pt}
  $$i:\N\vdash a\le i\le \mathrm{succ}^n(a)\to i=a\lor i=\operatorname{succ}(a)\lor\dots\lor i=\operatorname{succ}^n(a)\vspace{-6pt}$$
  is provable.
\end{theorem}
\begin{proof}
By induction on $n$. Applying the induction hypothesis with $\operatorname{succ}(a)$ and $n-1$, we get
$$i:\N\vdash \mathrm{succ}(a)\le i\le \mathrm{succ}^n(a)\to i=\operatorname{succ}(a)\lor\dots\lor i=\operatorname{succ}^n(a)$$
so it suffices to show $a\le i\to i=a\lor \mathrm{succ}(a)\le i$, and we appeal to lemma L.
\end{proof}

So ideally, when introducing $\le$ or the \texttt{REAL} requirement, one would be required to supply a proof of lemma L somehow to justify that the checker will be making use of this fact in the following article. Unfortunately, there is no actual place to inject this theorem into the system, because there is no concrete syntax for introducing requirements. That is, even taking all the \texttt{.miz} files in the MML together there is nothing that would indicate that \texttt{XXREAL\_0} is the article which allows the \texttt{REAL} requirement to be used.\footnote{Note that is also a Mizar article called \texttt{real.miz}, which contains some of the lemmas that are auto-proved by the \texttt{REAL} requirement, but there is no formal relation between the article and the requirement, and it is only an incomplete approximation to the lemmas required to justify the requirement, not formally checked. This should be easy to fix, and will more or less fall out of any attempt at proof export.}

The way this actually works is that when the accommodator sees a \lstinline|requirements REAL;| directive, it reads the (hand-written) \texttt{real.dre} XML file, which explicitly names the constructor number for $\le$ and the fact that it is in article \texttt{XXREAL\_0}. We would like to propose that this file also contains \emph{justifications} for involved constants so that the theorems aren't smuggled in without proof. (Or even better, requirement declarations become a part of the language proper, so that they can get \lstinline|correctness| proof blocks like any other justified property.)

The case of \texttt{NUMERALS + REAL} enabling flex-conjunctions is especially interesting because neither of these requirements depends on the other, so neither \texttt{numerals.dre} nor \texttt{real.dre} can state the compatibility theorem $a\le i\to i=a\lor \mathrm{succ}(a)\le i$ between them. (This doesn't require a major reorganization to fix, since of course the article that introduces \texttt{REAL}, \texttt{xxreal\_0.miz}, references the article which introduces \texttt{NUMERALS}, \texttt{ordinal1.miz}.)

\section{The analyzer}\label{sec:analyzer}

The analyzer plays an interesting role in Mizar. This is an essentially completely separate inference system from the checker, which has much stricter rules about equality of expressions, and it is what forms the ``glue'' between different lines of proof. It is a large module only because the language of Mizar is quite expansive, with 109 keywords, including:
\begin{itemize}
  \item Different kinds of definitions for modes, functors, predicates; redefinitions;
  \item Definitions by case analysis;
  \item ``Notations'' (\lstinline|synonym| and \lstinline|antonym| declarations);
  \item ``Properties'' like \lstinline|commutativity|, \lstinline|projectivity| or \lstinline|involutiveness|;
  \item Cluster registrations (existential, functor, and conditional);
  \begin{itemize}
    \item Existential clusters assert that $\overrightarrow{\chi}\ \tau$ is nonempty and hence a legal type
    \item Functor clusters assert that $t\is\overrightarrow{\chi}$ for some term $t$
    \item Conditional clusters assert that $\overrightarrow{\chi}$ implies $\overrightarrow{\chi}'$
  \end{itemize}
  \item ``Reductions'', equalities the checker automatically uses for simplification;
  \item ``Identifications'', equalities the checker automatically uses for congruence closure;
  \item Local declarations;
  \item Schemes;
  \item Reservations (variables with types declared in advance);
  \item Propositions and theorems;
\end{itemize}
and this is not an exhaustive list.

This also only covers top-level items. Inside a \lstinline|proof| block there is a different (overlapping but largely disjoint) set of legal items, called skeleton steps, which correspond to natural deduction rules. Inside a proof there is a variable that holds the current ``thesis'', the goal to prove, and the \lstinline|thesis| keyword resolves to it, unless one is in a \lstinline|now| block, where the thesis is not available and is reconstructed from the skeleton steps.

\begin{itemize}
  \item \lstinline|let x be T;| is the forall introduction rule: it transforms the thesis from $\forall x:T.\ \varphi(x)$ to $\varphi(x)$ and introduces a constant $x:T$.
  \item \lstinline|assume A;| is the implication introduction rule: it transforms the thesis from $A\to B$ to $B$ and pushes a proposition $A$, which can be referred to using \lstinline|then|.
  \item \lstinline|thus A;| is the conjunction introduction rule: it transforms the thesis from $A\land B$ to $B$, and gives $A$ as a goal to the checker.
  \item \lstinline|take t;| is the existential introduction rule: it transforms the thesis from $\exists x.\ \varphi(x)$ to $\varphi(t)$. There is also \lstinline|take x = t;| which is the same but introduces $x$ as an abbreviation for $t$; this version can also be used in a \lstinline|now| block.
  \item \lstinline|consider x being T such that A;| is existential elimination: it introduces $x:T$ and a proposition $A(x)$ that can be labeled, and gives $\exists x:T.\ A(x)$ as a goal to the checker.
  \item \lstinline|given x being T such that A;| is a combination of implication introduction and existential elimination: it transforms $(\exists x:T.\ A(x))\to B$ to $B$ and introduces $x:T$ and a proposition $A(x)$ that can be labeled.
  \item \lstinline|reconsider x = t as T;| introduces $x:T$ as a new local constant known to be equal to $t$, and gives $(t:T)$ as a proof obligation to the checker. (This is mainly used when $t:T$ is not already obvious to the type system, and allows $t:T$ to be proved by the user.)
  \item \lstinline|per cases;| is disjunction elimination, and it has two variations:
  \begin{itemize}
    \item followed by a sequence of \lstinline|suppose A_i; ... end;| blocks, it gives $\bigvee_iA_i$ as a goal to the checker and makes $A_i$ available in each block, leaving the thesis unchanged;
    \item followed by a sequence of \lstinline|case A_i; ... end;| blocks, it gives $\bigvee_iA_i$ as a goal to the checker, and if the thesis is $\bigvee_i(A_i\land B_i)$ then $B_i$ becomes the thesis in each block.
  \end{itemize}
\end{itemize}

Additionally, the formulas don't have to exactly match what the skeleton steps say. For example one can start a proof of $\forall x:\N.\ \varphi(x)$ using \lstinline|let x be set;| because \lstinline|set| is a supertype of \lstinline|Nat|. Definitional unfolding can also be forced by a skeleton step, for example if the thesis is $A\subseteq B$ then \lstinline|let x be set;| is a legal step provided the right \lstinline|definitions| directive is supplied.

The other major role of the analyzer is to elaborate types, terms, and formulas from their input form to the core grammar shown in section~\ref{sec:syntax}. There are two things that make this challenging:
\begin{itemize}
  \item Mizar heavily uses overloading, where the same function symbol can have several definitions (and \texttt{redefinition}s, which are definitions which use the same base term but can have different input and output types). These are resolved by declaration order (last declaration wins) and typing. Types are propagated exclusively from the inside out, using this type-based overload resolution, although you can use \texttt{e qua A} as a type ascription to influence the selection.
  \item Many declarations have ``invisible arguments'', also known as ``implicit arguments'' in the literature. These are filled in by a straightforward first-order unification process.\footnote{Unfortunately, the unification process does not have completely unique solutions, because attributes are unordered, and this can lead to overfitting in the MML to the Mizar attribute ordering. For example, $\texttt{rng}(f)\subseteq B$ returns the range of a function $f:A\to B$ as an element of \texttt{Subset of B} (and \texttt{rng} has hidden arguments $\texttt{rng}_{A,B}(f)$ inferred from the types). But the function type is split into separate attributes as \texttt{A-defined B-valued Function}, so if \texttt{f is B-valued C-valued Function} then it is up to a variety of implementation details whether you get the type of \texttt{rng f} as \texttt{Subset of B} or \texttt{Subset of C}. This happens in practice for the empty function, which is registered as \texttt{NAT-valued}, \texttt{RAT-valued}, and a few others.}
\end{itemize}

\section{Name resolution (MSM)}\label{sec:nameres}
The MSM (more strict Mizar) pass in the verifier is primarily concerned with name resolution and the \texttt{reserve} keyword. In \texttt{mizar-rs} this task is done as part of the analyzer, but in the original Mizar this is a separate pass, converting \emph{art}\texttt{.wsx} to \emph{art}\texttt{.msx}.

Generally speaking, name resolution is a straightforward process that proceeds as one would expect. A declaration like \lstinline|let x be set;| introduces a variable \texttt{x} until the end of the scope, which is a block such as \lstinline|proof ... end;|. Variables are locally bound by constructs like \lstinline|for x being set holds P[x]|.

Where things get interesting is when the top level command \lstinline|reserve x for set;| is used. This has two effects:
\begin{itemize}
  \item The variable \lstinline|x| can now be used as a binder without specifying the type: \lstinline|for x holds P[x]|. This would otherwise be an error. (Unlike some other type theory based systems, Mizar performs no type inference for variables based on usage.)
  \item The variable x will be implicitly quantified when used as a free variable in theorems. For example \lstinline|theorem x = x;| becomes legal and defines a theorem which says \lstinline|for x being set holds x = x|.
\end{itemize}

There is more to implicit quantification than meets the eye, though, because implicit quantification also occurs in nested positions. For example:
\begin{lstlisting}
reserve a,b,c,d for Nat;
{{c,a} : a + b < d} = d;
\end{lstlisting}
parses as:
\begin{lstlisting}
for d being Nat holds
  {{c,a} where a,c is Nat : for b being Nat holds a + b < d} = d;
\end{lstlisting}
The way we explain this behavior is by defining a tree of nested \emph{implicit quantification scopes} as follows:
\begin{itemize}
  \item The top level of a formula is an implicit quantification scope.
  \item A Fraenkel operator $\{e\mid \overline{x:T}\mid\varphi\}$ defines two implicit quantification scopes, one around the whole $\{e\mid \overline{x:T}\mid\varphi\}$ expression, and one around $\varphi$.
\end{itemize}
These scopes form a tree structure, and given a reserved variable we can mark each node according to whether it contains a use of the variable which is not in a sub-scope. For the above example we obtain the following marking structure:
\begin{center}
\vspace{1em}
\begin{tikzpicture}[auto]
  \node at (0,0) [fill=gray!30, shape=rounded rectangle] (main) {\lstinline|{{c,a} : a + b < d} = d|};
  \node at (0,-1) [fill=gray!30, shape=rounded rectangle] (fr) {\lstinline|{{c,a} : a + b < d}|};
  \node at (0,-2) [fill=gray!30, shape=rounded rectangle] (comp) {\lstinline|a + b < d|};
  \draw[->] (main) -- (fr);
  \draw[->] (fr) -- (comp);

  \node [right=2.4 of main.center] {\phantom{$a$}\quad\phantom{$b$}\quad\phantom{$c$}\quad\underline{$d$}};
  \node [right=2.4 of fr.center] {\underline{$a$}\quad\phantom{$b$}\quad\underline{$c$}\quad\phantom{$d$}};
  \node [right=2.4 of comp.center] {$a$\quad\underline{$b$}\quad\phantom{$c$}\quad$d$};
\end{tikzpicture}
\vspace{1em}
\end{center}
We then place the quantifier for each variable at every node such that none of its ancestors are marked. This means that \underline{$d$} is bound at the top, \underline{$a$} and \underline{$c$} are bound in the middle node, and \underline{$b$} is bound in the innermost node. A variable may end up bound more than once if it appears in disjoint subtrees. For example \lstinline|{0 : b = b} = {1 : b = b}| will bind $b$ twice to obtain \lstinline|{0 : for b being Nat holds b = b} = {1 : for b being Nat holds b = b}|.

Reserved variables can also have dependent types, such as \lstinline|reserve x for set,|\\ \lstinline|y for Element of x;|. In this case, a use of variable $y$ causes $x$ to also be considered as used and may cause it to be pulled in for implicit quantification. But this also interacts with the first use of reserved variables (type elision for binders), in that:
\begin{lstlisting}
for x holds {0 : y = y} = x
:: elaborates to ->
for x being set holds {0 : for y being Element of x holds y = y} = x

for z being set holds {0 : y = y} = z
:: elaborates to ->
for z being set holds
  {0 : for x being set, y being Element of x holds y = y} = z
\end{lstlisting}
where we note that the binder for $y$ has a type referring to the explicit \lstinline|for x| binder in the first example. Because implicit quantifications always go first, before any explicit quantifiers, there are ways to use this to obtain invalid expressions, for example \lstinline|for x holds y = y| would elaborate to \lstinline[mathescape]|for y being Element of $\color{red}\texttt{x}$, x being set holds y = y|, where $y$'s type attempts to refer to a binder deeper in the expression. Such constructs are rejected by the checker.\footnote{At least, \texttt{mizar-rs} rejects these constructs. Original Mizar appears not to be hardened against this, and gives an assertion violation with this example.}

\section{The parser}\label{sec:parser}
Working from back to front, we finally end up at the parser, which is the component which reads the input file and produces an abstract syntax tree (AST) which is processed by the stages we have already seen. Several papers about the Mizar parser have been written, inlcuding \cite{bylinski2012,cairns200460,nakasho2019}.

\subsection{The scanner}\label{sec:scanner}
The first component in the parser is the scanner, also known in the literature as the tokenizer. This is responsible for separating the input text into a sequence of ``tokens'' which are then built into a tree structure in the parser. This is a fairly standard design for a programming language, but the Mizar tokenizer makes some unique choices which have historically make parsing Mizar tricky.

The top of a Mizar file is the environment section of the form \lstinline|environ ... begin|, which is pre-parsed with a fixed grammar. The accommodator reads only this part of the file, and then loads all the data for the referenced articles, before the parser can proceed to the rest of the file. This is important, because the Mizar tokenizer is extensible. In contrast to some ITP systems (and to the frustration of its users), tokens are defined by an external \emph{art}\texttt{.voc} file, or in the global \texttt{mizar.voc} file used for definitions in the MML, rather than in the \emph{art}\texttt{.miz} file itself. This means that:
\begin{lstlisting}
definition
  let x be Nat;
  func foo x -> Nat equals x+1;
  correctness;
end;
\end{lstlisting}
results in some parse errors like \texttt{Functor symbol expected} at the \lstinline|foo| token, which are fixed by adding the line \lstinline|Ofoo| to the \emph{art}\texttt{.voc} file, where the \texttt{O} means that this is a functor symbol declaration.

This usability issue seems relatively easy to fix---the definition site \lstinline|func foo x -> ...| could be taken to define a new token \lstinline|foo|, rather than having to be declared separately (in a not particularly human readable file). But \texttt{mizar-rs} currently just follows the same strategy as PC Mizar here.

\subsection{The parser}\label{sec:parser-2}
The Mizar language has a rather large grammar, with 115 keywords like \lstinline|assume|, \lstinline|hereby|, \lstinline|involutiveness|, and \lstinline|reducibility|. This derives from Mizar's origin as a controlled natural language, where Mizar proofs ``read like mathematics'' and the grammar does not need to be learned explicitly. On the other hand, Mizar does not have very many synonyms, in the sense of keywords that can be freely interchanged, so when writing Mizar one very much does need to know the formal grammar.

We have already seen the top level constructs of a Mizar file in \autoref{sec:analyzer}. The formal grammar is also presented in BNF form on the Mizar website\footnote{\url{https://mizar.uwb.edu.pl/language/}}, so we will not dwell on it. However, \texttt{mizar-rs} does not (and cannot) simply use a parser generator applied to the BNF to get a working Mizar parser. There are a few reasons for this:
\begin{itemize}
  \item This makes it possible to implement the parser without a lot of lookahead: the maximum lookahead is two tokens, to differentiate \lstinline|A: 1 = 1| from \lstinline|A + 1 = 1|. Using the BNF directly would result in unbounded lookahead.
  \item There are also some constraints which would bloat the BNF; for example, the term former \lstinline[mathescape]|the $\mbox{\emph{field}}$| is legal only within a structure declaration; outside a structure one must instead use \lstinline[mathescape]|the $\mbox{\emph{field}}$ of $\mbox{\emph{struct}}$|. One would need to have two copies of the term grammar to express this in a context-free way.
  \item There are some cases of genuine ambiguity in the grammar.
  \begin{itemize}
    \item The BNF defines \lstinline[mathescape]|$\varphi$ & $\psi$ or $\chi$| without a precedence ordering, which would imply two valid parses. In reality there is a precedence ordering here and only one parse is valid. There are no user-defined formula operators, so this is a fixed list.
    \item Similarly, infix operators in the term grammar have defined precedences, so that \lstinline[mathescape]|$a$ + $b$ * $c$| parses as expected. In this case the operators are user-defined, so precedence is specified in the \emph{art}\texttt{.voc} file and the parser has to handle a dynamic-precedence grammar.
    \item One particularly nasty case of ambiguity in the grammar is caused by the fact that commas are used for separating between binder groups, and also for separating the arguments of a type. For example:
\begin{lstlisting}
reserve x for set, y for Element of x, f,g for Function of x,y, w for set;
\end{lstlisting}
    We have inserted extra spaces to clarify, but Mizar has to know the ``arity'' of \lstinline|Element| and \lstinline|Function| here to realize that in ``\lstinline|x, f,g|'', the \lstinline|x| is part of the type of \lstinline|y| while \lstinline|f,g| are new variable directives, and in ``\lstinline|x,y, w|'' the first two variables are arguments because \lstinline|Function of x,y| takes two arguments.

    To make things worse, a type operator need not have only one arity. For example \lstinline|Relation of X| and \lstinline|Relation of X,Y| are both defined (meaning subsets of $X\times X$ and $X\times Y$ respectively). Mizar (and \texttt{mizar-rs}) currently uses a heuristic to avoid backtracking in this situation, which is to take all arguments up to the maximum arity of the type operator.
  \end{itemize}
\end{itemize}

Disambiguating infix expressions in the term grammar turns out to be especially subtle and a bit mathematically interesting, so we will go into a bit more detail there. Mizar allows infix functions with more than 2 arguments, by using a tuple-like notation before or after the symbol. In general, a functor application has the form $(x_1,\dots,x_m)\;F\;(y_1,\dots,y_m)$, where $F$ is a functor symbol, and there are $m$ left and $n$ right arguments to the function. When $m=1$ the parentheses are optional, and when $m=0$ nothing is written. So this generalizes the notion of prefix, postfix, and infix functions like \lstinline|F(x,y,z)|, \lstinline|x + y|, and \lstinline|x ^2|. But just like with type constructors, functor symbols can be overloaded, and in fact the parser must disambiguate based on $(F,m,n)$, that is, it resolves a functor symbol together with its number of left and right arguments, so you could simultaneously have \lstinline|-y|, \lstinline|x-y|, and \lstinline|x-(y,z)| as independent operations associated to the functor symbol \lstinline|-|.

In general, the input to precedence resolution is a sequence of operators separated by arguments, of the form:
$$(x_{01},\dots,x_{0k_0})\;F_1\;(x_{11},\dots,x_{1k_1})\;F_2\;\dots \;F_n\;(x_{n1},\dots,x_{nk_n}).$$
The goal is to produce a suitable parenthesization of the operators (that is, a binary tree with $F_i$ on internal nodes and $(x_{ij})_{j\le k_i}$ on the leaves) respecting precedence such that each $F_i$ is applied to the right number of arguments. One might think that we can just put the operators into a tree such that $F_i$ is a left child of $F_j$ if $\operatorname{prec}(F_i)\le \operatorname{prec}(F_j)$ and $F_i$ is a right child of $F_j$ if $\operatorname{prec}(F_i)<\operatorname{prec}(F_j)$ (note, all precedence levels in Mizar are left-associative), but this produces an incorrect result in practice.

An example of this is the expression \lstinline|union union X|, which we can pre-parse as\\ $()\;\texttt{union}\;()\;\texttt{union}\;(X)$. Grouping this only using precedence and left-associativity results in the parse $(()\;\texttt{union}\;())\;\texttt{union}\;(X)$, where the first $\texttt{union}$ operator is applied to 0 left and right arguments, and the second is applied to 1 left and 1 right argument. In this case neither operator exists and we give an error, but the alternative parse $()\;\texttt{union}\;(()\;\texttt{union}\;(X))$ would apply both $\texttt{union}$ operators to 0 left and 1 right argument, which does exist (\lstinline|union X| is the union of a set of sets).

We can formalize this as a constrained optimization problem, where we want to pick a solution which violates as few precedence ordering constraints as possible, among solutions that only use operator formats that actually exist. The unconstrained precedence ordering problem can be solved in $O(n)$ in the length of the chain of operators, but to our knowledge there is no similar $O(n)$ algorithm for constrained precedence parsing. There does exist an $O(n^4)$ dynamic programming algorithm, in which we calculate for each $1\le a\le i \le b\le n$ on the substring $F_a,\dots,F_b$ of the original problem the minimal cost (or $+\infty$) of a solution having $F_i$ at the root of the tree. (There are $O(n^3)$ of these subproblems and they can be calculated from smaller subproblems in $O(n)$.)

Original Mizar does not actually use this algorithm, it uses an approximation of it which is $O(n^2)$. So (unfortunately) \texttt{mizar-rs} follows suit, since MML is already known to satisfy it. It would be interesting to see whether there exists a hybrid algorithm which works in the general case and is still $O(n^2)$ in the situation where the original Mizar algorithm works.

\section{Mizar soundness bugs}\label{sec:soundness-bugs}

One of the fortuitous side effects of going over each line of code and rewriting it to something morally equivalent in a different language is that one can find a lot of bugs. Bugs can happen even when one takes great efforts to avoid them \cite{adams2016proof,kunvcar2015consistent}, but external review can definitely help. Mizar is a large project with a long history and a small team, whose source code was not publicly accessible, with many soundness-critical parts, which is pretty much a worst case scenario for finding soundness bugs. This \texttt{mizar-rs} project has been tremendously successful in ferreting out these bugs, with no less than four proofs of false that will be given below. We reported these bugs to the Mizar developers, and they have been fixed (and the MML patched) in version 8.1.14.

While it is unfortunate that the software wasn't perfect to start with, this is evidence for the usefulness of external checkers, and it is a way for us to improve the original Mizar. We hope that by telling the story of how these bugs work we can give some sense of some of the issues that can arise when doing proof checking, as well as some more internal details whose importance may not have been obvious.

\subsection{Exhibit 1: polynomial arithmetic overflow}\label{sec:bug1}

This is the largest of the contradiction proofs, and we will show only the main part of it.\footnote{The full proof can be found at\\ \url{https://github.com/digama0/mizar-rs/blob/itp2023-2/itp2023/false1/false1.miz}.} The only non-MML notion used is the adjective \lstinline|a is x-ordered| defined as $x<a$.
\begin{lstlisting}
theorem contradiction
proof
  consider x being 1-ordered Nat such that not contradiction;
  1 is 0-ordered; then
  A1: x * x is 1-ordered; then
  consider x1 being 0-ordered Nat such that B1: x1 = x * x;
  A2: 1 < x1 by A1,B1; then
  x1 * x1 is x1-ordered by XREAL_1:155; then :: 0 < a & 1 < b -> a < a*b
  consider x2 being x1-ordered Nat such that B2: x2 = x1 * x1;
  ...
  consider x31 being x1-ordered Nat such that B31: x31 = x30 * x30;
  consider x32 being x1-ordered Nat such that B32: x32 = x31 * x31;
  C: x32 * x1 = x1 by
    B1,B2,B3,B4,B5,B6,B7,B8,B9,B10,B11,B12,B13,B14,B15,B16,
    B17,B18,B19,B20,B21,B22,B23,B24,B25,B26,B27,B28,B29,B30,B31,B32;
  x32 is x1-ordered implies x1 < x32;
  then 0 < x1 & 1 < x32 by A2,XXREAL_0:2; :: a <= b & b <= c -> a <= c
  hence contradiction by C,XREAL_1:155;   :: 0 < a & 1 < b -> a < a*b
end;
\end{lstlisting}

As mentioned earlier, there is a module in the equalizer for polynomial evaluation. This means that each expression which is a complex number will also be expressed as a polynomial in terms of basic variables, and two expressions which compute to equal polynomials will be equated. This is how things like $-(a+b)=-a+-b$ are proved automatically by the checker.

These polynomials are of the form $\sum_ic_i \prod_jx_{i,j}^{n_{i,j}}$, and the starting point for this proof was the discovery that the $n_{i,j}$ are represented as signed 32-bit integers and overflow is not checked. Turning this into an exploit is surprisingly difficult however, because one cannot simply write down $x^{2^{32}}$ because the power function is not one of the requirements (see section \ref{sec:requirements}). The best thing we have for creating larger polynomials is multiplication, but we can get to $x^{2^{32}}$ with 32 steps of repeated squaring, which is what the large block of \lstinline|consider| statements is doing.

So the strategy is to construct $x_1=x^2$, $x_2=x^4$, all the way up to $x_{30}=x^{2^{30}}$. After this things start to get weird: $x_{31}=x^{-2^{31}}$ because of signed overflow, and $x_{32}=x^{0}$. The system does not recognize $x_{32}=1$ however, because it maintains an invariant of monomial powers being nonzero, and single powers are also handled specially, so we have to go up to $x_{32}\cdot x_1=x^{2}$. Since $x_1$ is also $x^2$, the system will equate $x_{32}\cdot x_1=x_1$ which is the key step \textbf{C}. After this, we simply need to separately prove that since $x_{32}\cdot x_1$ is really $x^{2^{32}+2}$, it is strictly larger than $x_1=x^2$ as long as we choose $x>1$.

To prove this last fact we use the attribute inference mechanism to prove that all of the intermediates are strictly greater than $x_1$, since $x_1<a$ implies $x_1<a\cdot a$ as long as $x_1>1$. Hence $1<x_1<x_{32}$ so $x_1<x_{32}\cdot x_1=x_1$, which is a contradiction.

\subsection{Exhibit 2: negation in the schematizer}\label{sec:bug2}

This one requires absolutely no imports; we include the complete proof below including the import section.
\begin{lstlisting}
environ begin

scheme Foo{P[set,set]}: P[1,1] implies P[1,1]
proof thus thesis; end;

theorem contradiction
proof
  1 = 1 implies 1 <> 1 from Foo;
  hence thesis;
end;
\end{lstlisting}

This defines a very trivial scheme which just says that $P(1,1)$ implies $P(1,1)$. The interesting part is the instantiation of this scheme in the main proof. In the schematizer, we do not negate the thesis and try to prove false, we just directly match the thesis against the goal. It does not attempt any fancy higher-order unification: if it needs to unify $P(\overrightarrow{t_i})\overset{?}{=}\pm R(\overrightarrow{t_i}')$ it will just assign $P:=\pm R$ and proceed with $t_i\overset{?}{=}t_i'$ for each $i$. Or at least it should do that, but there is a bug wherein it instead assigns $P:=R$ and ignores the $\pm$ part, so we can trick it into unifying $P(x,y):=(x=y)$ and then think that $P(x,y)\overset{?}{=}(x\ne y)$ is still true. So we use $1=1$ to prove $1\ne 1$ and then prove a contradiction.

\subsection{Exhibit 3: flex-and unfolding}\label{sec:bug3}
\begin{lstlisting}
theorem contradiction
proof
  (1 = 1 or 1 = 1) & ... & (2 = 1 or 1 = 1);
  hence thesis;
end;
\end{lstlisting}

This one is bewilderingly short. We start by asserting a (true) flex-and statement $\bigwedge_{i=1}^2(i=1\lor 1=1)$. We prove this by using the forall expansion of the flex-and,\\ $\forall i:\N.\ 1\le i\le 2\to i=1\lor 1=1$, which is of course true because the $1=1$ disjunct is provable.

The second part is simply \lstinline|hence contradiction|, so we are calling the checker to \emph{disprove} the same statement. So we assume $\bigwedge_{i=1}^2(i=1\lor 1=1)$ and one of the things the pre-checker does here is to expand out the conjunction, and we would expect it to produce $(1=1\lor 1=1)\land (2=1\lor 1=1)$ from which no contradiction can be found. However, what it actually produces is $1=1\land 2=1$, which can be disproved.

To see why this happens, we have to look more specifically at the forall-expansion without the negation sugar employed thus far. What Mizar actually sees for the expansion of the flex-and expression is $\forall i:\N.\ \neg(1\le i\wedge i\le 2\wedge (i\ne 1\wedge 1\ne 1))$, except that as mentioned previously conjunctions are always flattened into their parents. The code for expanding flex-and expects the expansion body to be the third conjunct past the forall and negation, but after flattening the third conjunct is actually $i\ne 1$ and there is an unexpected fourth conjunct $1\ne 1$ that is forgotten.

\subsection{Exhibit 4: flex-and substitution}\label{sec:bug4}

This is the most worrisome of the bugs that have been presented, because it is not simply a bad line of code but rather an issue with the algorithm itself. As a result, this one survived the translation to Rust and was discovered to affect both versions, and moreover the fix could not be rolled out without breaking the MML. The MML maintainers have stepped up and fixed the proofs though after we brought this issue to their attention, so we were able to successfully excise this bug from both Mizar implementations.

\begin{lstlisting}
theorem contradiction
proof
  A: now
    let n be Nat;
    assume n > 3 & (1 + 1 <> 3 & ... & n <> 3);
    hence contradiction;
  end;
  ex n being Nat st n > 3 & (1 + 1 <> 3 & ... & n <> 3)
  proof
    take 2 + 2;
    thus 2 + 2 > 3 & (1 + 1 <> 3 & ... & 2 + 2 <> 3);
  end;
  hence thesis by A;
end;
\end{lstlisting}

To explain what is happening here, we have to talk about another aspect of flex-and expressions which was not mentioned in section \ref{sec:syntax}, which is that a flex-and expression $\bigwedge_{i=a}^b\varphi(i)$ is actually stored as five pieces of information: the bounds $a$ and $b$, the expansion $\forall i:\N.\ a\le i\le b\to \varphi(i)$, from which $\varphi(i)$ can be reconstructed (if one is careful -- see section \ref{sec:bug3}), and the bounding expressions $\varphi(a)$ and $\varphi(b)$. As the reader may have noticed, the syntax of Mizar very much prefers to only discuss the bounding expressions, since flex-and expressions in the concrete syntax are written as $\varphi(a)\wedge\dots\wedge\varphi(b)$ with literal dots, and no place to supply $a$, $b$ or $\varphi(i)$.

When one writes an expression like $P\wedge \dots \wedge Q$, Mizar essentially diffs the two expressions to determine what is changing and what the values are on each side. So if one writes $1+1\ne 3\wedge \dots\wedge n\ne 3$ then Mizar infers that the desired expression is $\bigwedge_{i=1+1}^n i\ne 3$.

The problem is that this operation is not stable under substitution. If we take that expression and substitute $n:=2+2$, then we end up with $\bigwedge_{i=1+1}^{2+2} i\ne 3$, but if we were to write out $1+1\ne 3\wedge \dots\wedge 2+2\ne 3$ ourselves we would end up with $\bigwedge_{i=1}^{2} i+i\ne 3$ instead. This is somewhat annoying, especially if one is trying to write an exploit example, but it's not obviously a soundness bug yet because internally we are still storing the expansion which has all the details -- $\bigwedge_{i=1+1}^{2+2} i\ne 3$ and $\bigwedge_{i=1}^{2} i+i\ne 3$ are different expressions.

However, the equality check between two flex-and expressions is simply $\varphi(a)=\varphi'(a)$ and $\varphi(b)=\varphi'(b)$! The example above is crafted to demonstrate that this is not sound in general, since $\bigwedge_{i=1+1}^{2+2} i\ne 3$ is false but $\bigwedge_{i=1}^{2} i+i\ne 3$ is true.

In the first part of the proof we show that $\forall n.\ \neg(n>3\wedge \bigwedge_{i=1+1}^{n} i\ne 3)$, which is true. The interesting part is the second half, where we take $n:=2+2$. At this point the thesis is $2+2>3\wedge \bigwedge_{i=1+1}^{2+2} i\ne 3$, but we cannot write this expression. What we write instead is $2+2>3\wedge (1+1\ne 3\wedge \dots\wedge 2+2\ne 3)$, which as mentioned will be interpreted as $2+2>3\wedge\bigwedge_{i=1}^{2} i+i\ne 3$. This should be rejected for not matching the thesis, but it has the same endpoints so it is accepted, and the checker is able to prove it by case analysis. If we were to write \lstinline|thus thesis| to let Mizar insert the unmentionable proposition, then the analyzer will accept it as being the right thesis but the checker will not be able to prove it.

\subsection{Honorable mention: attributes that don't exist}\label{sec:bug5}

This is not an exploitable bug to my knowledge, but it is notable for being widespread, and it is difficult for \texttt{mizar-rs} to support without resulting in very weird behavior.\footnote{It is also the only bug among those discussed here which is not fixed in the latest version (8.1.14).} Consider the following Mizar article:
\begin{lstlisting}
environ
  vocabularies ZFMISC_1, SUBSET_1;
  notations ZFMISC_1, SUBSET_1;
  constructors TARSKI, SUBSET_1;
  requirements SUBSET; ::, BOOLE;
begin
for x,B being set, A being Element of bool B st x in A holds x in B;
\end{lstlisting}
This checks as one would expect, and it is in fact a true statement (and it is not at all contrived). However, the reasoning that gets the checker to accept the proof is\ldots\ somewhat suspicious. It goes as follows:
\begin{enumerate}
  \item Suppose $A:\operatorname{Element}(\mathcal{P}(B))$, $x\in A$ and $x\notin B$.
  \item Because $A:\operatorname{Element}(\mathcal{P}(B))$ and $x\in A$, it follows that $B$ is not empty and $x:\operatorname{Element}(B)$.
  \item Because $x\notin B$, $B$ is not empty, and $x:\operatorname{Element}(B)$, contradiction.
\end{enumerate}
Because the example has been minimized, there is really not much going on in the proof because these are all essentially primitive inferences. The problem here is ``$B$ is not empty'', because the \texttt{BOOLE} requirement which supplies the \lstinline|B is not empty| predicate is not available (note the environment section). The constructor for empty is actually available in the environment because it has been brought in indirectly via the \lstinline|constructors SUBSET_1| directive, but without the requirement the checker just sees it as a normal attribute.

So what, then, is the checker doing? This seems to be a case of multiple bugs cancelling each others' effects. Requirements are internally represented by an integer index, where zero means that the requirement is not available. Normally any handling of a requirement involves a check that the requirement is nonzero first, but we forget that in step 2, and as a result we end up adding attribute $0$ to $B$, which is meaningless. The second bug is in step 3, where we again forget to ask whether the requirement index is nonzero, and so we find attribute $0$ and interpret it (correctly) as meaning that $B$ is not empty.

This would just be a curious bug, but for the fact that it is exploited all over the place because when creating articles it is standard to minimize the environment section by removing anything that keeps the proof valid, and this bug allows one to remove the \texttt{BOOLE} requirement without breaking the proof. We had to patch 46 articles that all had this same issue. In all cases we only need to add the \texttt{BOOLE} requirement to fix the issue.

\section{Results}\label{sec:results}

\begin{wraptable}{r}{0.25\textwidth}
  \vspace{-16mm}
  \caption{Line counts by category in \texttt{mizar-rs}.}
\label{fig:linecount}
\begin{center}
\begin{tabular}{lr}
  \toprule
  category & lines \\ \midrule
  accom. & 613 \\
  parser & 2781 \\
  analyzer & 6353 \\
  checker & 13631 \\
  general & 5214 \\ \midrule
  total & 28592 \\
  \bottomrule
\end{tabular}
\end{center}

  \vspace{-5mm}
\end{wraptable}

\begin{figure}[tb]
  \centering
  \includegraphics[width=\columnwidth]{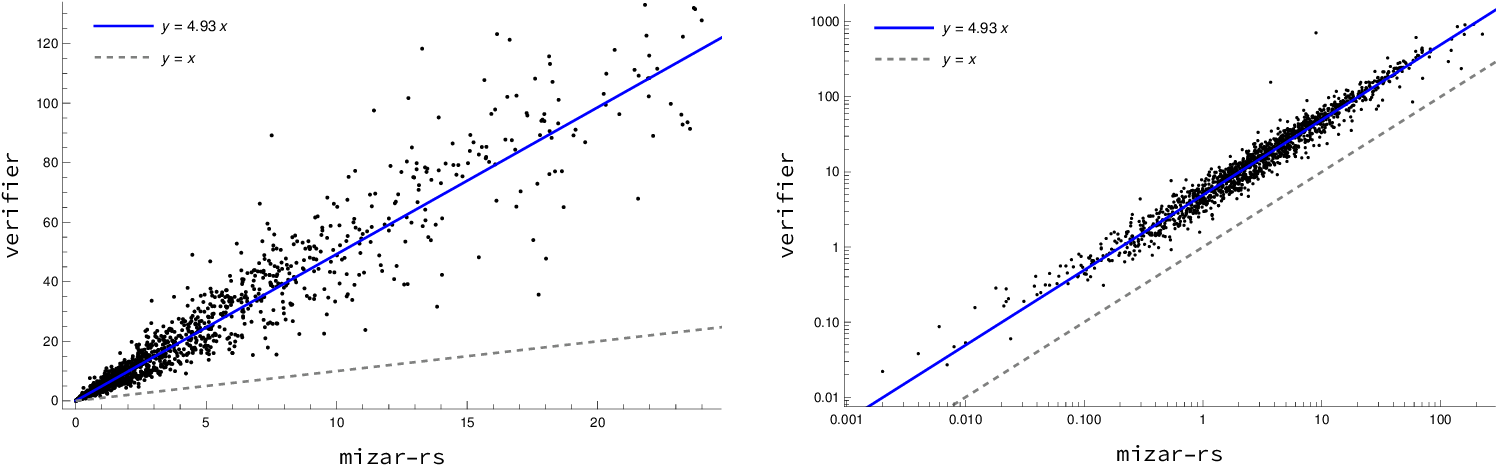}
  \caption{Scatter plot of files in the MML based on check time in \texttt{verifier} (vertical axis) vs. in \texttt{mizar-rs} (horizontal axis), in seconds. The left plot is linear and the right is log-log (including some outliers not in the left plot). The {\color{blue}blue} line is the best fit, indicating an average $4.93\times$ speedup, and the gray dashed line is a $1\times$ reference; there are no points below this line, indicating that \texttt{mizar-rs} was never slower than \texttt{verifier} on this run.}
  \vspace{-5pt}\label{fig:checker-files}
\end{figure}

The whole project is 28\,592 lines of code (see figure \ref{fig:linecount}), or 13\,631 if we restrict attention to those files used in the checker, which is 1/2 to 1/3 of the equivalent code in Pascal (see figure \ref{fig:comparison}). We credit this mainly to the language itself -- Rust is able to express many complex patterns that would be significantly more verbose to write in Pascal, and PC Mizar is also written in a fairly OOP-heavy style that necessitates a lot of boilerplate. Furthermore, Pascal is manually memory-managed while Rust uses ``smart-pointer'' style automatic memory management, so all the destructors simply don't need to be written which decreases verbosity and eliminates many memory bugs.

The performance improvements are most striking: we measured almost 5$\times$ reduction in processing time to check the MML. See \autoref{fig:checker-files}. This is most likely a combination of characteristics of the LLVM compiler pipeline, along with many small algorithmic improvements and removing redundant work. (The reduction is even larger when adding the other phases of the Mizar system -- accom, parser, MSM -- so that we can skip the costly I/O and serialization steps of section \ref{sec:user-internals}.)

\begin{table*}[tb]
  \caption{Comparison of the original (Pascal) implementation of Mizar with \texttt{mizar-rs} (Rust). Tests were performed on a 12 core 12th Gen Intel Core i7-1255U @ 2.20GHz, with 12 threads. The times are measured in CPU minutes, and the ``par.'' column gives the observed parallelism as a ratio between real time and CPU time (e.g. a full \texttt{mizar-rs} run takes $\frac{167.58}{10.1}=16.62$ min of real time), which is roughly constant because all tasks have ample concurrency. (Note that \texttt{prel/} is precalculated before the runs, so that all articles can be processed in parallel.)}
\label{fig:comparison}
\begin{tabular*}{\textwidth}{@{\extracolsep\fill}cccccc}
  \toprule
  & \multicolumn{2}{@{}c@{}}{PC Mizar} & \multicolumn{2}{@{}c@{}}{\texttt{mizar-rs}} \\\cmidrule{2-3}\cmidrule{4-5}%
  Enabled passes & CPU time & par. & CPU time & par. & ratio \\ \midrule
  \phantom{accom +} analyzer \phantom{+ checker} & 129.90 min & 9.70 & 15.34 min & 9.54 & $8.47\times$ \\
  accom + analyzer \phantom{+ checker} & 141.43 min & 9.35 & 13.80 min & 10.3 & $10.2\times$ \\
  \phantom{accom + analyzer +} checker & 698.58 min & 10.2 & 142.75 min & 9.69 & $4.89\times$ \\
  \phantom{accom +} analyzer + checker & 748.40 min & 10.5 & 165.00 min & 10.2 & $4.54\times$ \\ \midrule
  accom + analyzer + checker & 798.38 min & 9.55 & 167.58 min & 10.1 & $4.76\times$ \\
  \bottomrule
\end{tabular*}

\end{table*}

\section{Conclusion \& Future work}

We have implemented a drop-in replacement for the \texttt{verifier} checker of the Mizar system that is able to check the entire MML, which gets a significant performance improvement. Furthermore, we have improved Mizar by uncovering and reporting some soundness bugs.

While this implementation is explicitly trying not to diverge from the Mizar language as defined by the PC Mizar implementation and the MML, there are many possible areas where improvements are possible to remove unintentional or undesirable restrictions. For example, anyone who has played with Mizar will have undoubtedly noticed that all the article names are 8 letters or less, for reasons baked deeply into PC Mizar. Should \texttt{mizar-rs} become the official Mizar implementation, it would be possible to lift this restriction without much difficulty.

This project was also originally started to get proof export from Mizar, and to that end replacing one large trusted tool with another one does not seem like much of an improvement. But (bugs notwithstanding) we saw nothing while auditing the checker that is not proof-checkable or unjustified, and remain confident that proof export is possible.


\bibliography{references}


\begin{thebibliography}{15}
\ifx \bisbn   \undefined \def \bisbn  #1{ISBN #1}\fi
\ifx \binits  \undefined \def \binits#1{#1}\fi
\ifx \bauthor  \undefined \def \bauthor#1{#1}\fi
\ifx \batitle  \undefined \def \batitle#1{#1}\fi
\ifx \bjtitle  \undefined \def \bjtitle#1{#1}\fi
\ifx \bvolume  \undefined \def \bvolume#1{\textbf{#1}}\fi
\ifx \byear  \undefined \def \byear#1{#1}\fi
\ifx \bissue  \undefined \def \bissue#1{#1}\fi
\ifx \bfpage  \undefined \def \bfpage#1{#1}\fi
\ifx \blpage  \undefined \def \blpage #1{#1}\fi
\ifx \burl  \undefined \def \burl#1{\textsf{#1}}\fi
\ifx \doiurl  \undefined \def \doiurl#1{\url{https://doi.org/#1}}\fi
\ifx \betal  \undefined \def \betal{\textit{et al.}}\fi
\ifx \binstitute  \undefined \def \binstitute#1{#1}\fi
\ifx \binstitutionaled  \undefined \def \binstitutionaled#1{#1}\fi
\ifx \bctitle  \undefined \def \bctitle#1{#1}\fi
\ifx \beditor  \undefined \def \beditor#1{#1}\fi
\ifx \bpublisher  \undefined \def \bpublisher#1{#1}\fi
\ifx \bbtitle  \undefined \def \bbtitle#1{#1}\fi
\ifx \bedition  \undefined \def \bedition#1{#1}\fi
\ifx \bseriesno  \undefined \def \bseriesno#1{#1}\fi
\ifx \blocation  \undefined \def \blocation#1{#1}\fi
\ifx \bsertitle  \undefined \def \bsertitle#1{#1}\fi
\ifx \bsnm \undefined \def \bsnm#1{#1}\fi
\ifx \bsuffix \undefined \def \bsuffix#1{#1}\fi
\ifx \bparticle \undefined \def \bparticle#1{#1}\fi
\ifx \barticle \undefined \def \barticle#1{#1}\fi
\bibcommenthead
\ifx \bconfdate \undefined \def \bconfdate #1{#1}\fi
\ifx \botherref \undefined \def \botherref #1{#1}\fi
\ifx \url \undefined \def \url#1{\textsf{#1}}\fi
\ifx \bchapter \undefined \def \bchapter#1{#1}\fi
\ifx \bbook \undefined \def \bbook#1{#1}\fi
\ifx \bcomment \undefined \def \bcomment#1{#1}\fi
\ifx \oauthor \undefined \def \oauthor#1{#1}\fi
\ifx \citeauthoryear \undefined \def \citeauthoryear#1{#1}\fi
\ifx \endbibitem  \undefined \def \endbibitem {}\fi
\ifx \bconflocation  \undefined \def \bconflocation#1{#1}\fi
\ifx \arxivurl  \undefined \def \arxivurl#1{\textsf{#1}}\fi
\csname PreBibitemsHook\endcsname

\bibitem[\protect\citeauthoryear{Bancerek et~al.}{2015}]{bancerek2015mizar}
\begin{bchapter}
\bauthor{\bsnm{Bancerek}, \binits{G.}},
\bauthor{\bsnm{Byli{\'n}ski}, \binits{C.}},
\bauthor{\bsnm{Grabowski}, \binits{A.}},
\bauthor{\bsnm{Korni{\l}owicz}, \binits{A.}},
\bauthor{\bsnm{Matuszewski}, \binits{R.}},
\bauthor{\bsnm{Naumowicz}, \binits{A.}},
\bauthor{\bsnm{Pak}, \binits{K.}},
\bauthor{\bsnm{Urban}, \binits{J.}}:
\bctitle{{Mizar: State-of-the-art and beyond}}.
In: \bbtitle{Intelligent Computer Mathematics: International Conference, CICM
  2015, Washington, DC, USA, July 13-17, 2015, Proceedings.},
pp. \bfpage{261}--\blpage{279}
(\byear{2015}).
\bcomment{Springer}
\end{bchapter}
\endbibitem

\bibitem[\protect\citeauthoryear{Bancerek et~al.}{2018}]{bancerek2018role}
\begin{barticle}
\bauthor{\bsnm{Bancerek}, \binits{G.}},
\bauthor{\bsnm{Byli{\'n}ski}, \binits{C.}},
\bauthor{\bsnm{Grabowski}, \binits{A.}},
\bauthor{\bsnm{Korni{\l}owicz}, \binits{A.}},
\bauthor{\bsnm{Matuszewski}, \binits{R.}},
\bauthor{\bsnm{Naumowicz}, \binits{A.}},
\bauthor{\bsnm{P{\k{a}}k}, \binits{K.}}:
\batitle{{The role of the Mizar Mathematical Library for interactive proof
  development in Mizar}}.
\bjtitle{Journal of Automated Reasoning}
\bvolume{61},
\bfpage{9}--\blpage{32}
(\byear{2018})
\end{barticle}
\endbibitem

\bibitem[\protect\citeauthoryear{Grabowski et~al.}{2010}]{grabowski2010mizar}
\begin{barticle}
\bauthor{\bsnm{Grabowski}, \binits{A.}},
\bauthor{\bsnm{Korni{\l}owicz}, \binits{A.}},
\bauthor{\bsnm{Naumowicz}, \binits{A.}}:
\batitle{Mizar in a nutshell}.
\bjtitle{Journal of Formalized Reasoning}
\bvolume{3}(\bissue{2}),
\bfpage{153}--\blpage{245}
(\byear{2010})
\end{barticle}
\endbibitem

\bibitem[\protect\citeauthoryear{Muzalewski}{1993}]{muzalewski1993outline}
\begin{bbook}
\bauthor{\bsnm{Muzalewski}, \binits{M.}}:
\bbtitle{{An Outline of PC Mizar}}.
\bpublisher{Fondation Philippe le Hodey}, \blocation{???}
(\byear{1993})
\end{bbook}
\endbibitem

\bibitem[\protect\citeauthoryear{Carneiro}{2022}]{PRIMRECI}
\begin{barticle}
\bauthor{\bsnm{Carneiro}, \binits{M.}}:
\batitle{{The Divergence of the Sum of Prime Reciprocals}}.
\bjtitle{Formalized Mathematics}
\bvolume{30}(\bissue{3}),
\bfpage{209}--\blpage{210}
(\byear{2022})
\doiurl{10.2478/forma-2022-0015}
\end{barticle}
\endbibitem

\bibitem[\protect\citeauthoryear{Matuszewski and
  Rudnicki}{2005}]{matuszewski2005mizar}
\begin{barticle}
\bauthor{\bsnm{Matuszewski}, \binits{R.}},
\bauthor{\bsnm{Rudnicki}, \binits{P.}}:
\batitle{Mizar: the first 30 years}.
\bjtitle{Mechanized mathematics and its applications}
\bvolume{4}(\bissue{1}),
\bfpage{3}--\blpage{24}
(\byear{2005})
\end{barticle}
\endbibitem

\bibitem[\protect\citeauthoryear{Harrison}{1996}]{harrison-mizar}
\begin{bchapter}
\bauthor{\bsnm{Harrison}, \binits{J.}}:
\bctitle{{A Mizar Mode for HOL}}.
In: \beditor{\bsnm{Wright}, \binits{J.}},
\beditor{\bsnm{Grundy}, \binits{J.}},
\beditor{\bsnm{Harrison}, \binits{J.}} (eds.)
\bbtitle{Theorem Proving in Higher Order Logics: 9th International Conference,
  TPHOLs'96}.
\bsertitle{Lecture Notes in Computer Science},
vol. \bseriesno{1125},
pp. \bfpage{203}--\blpage{220}.
\bpublisher{Springer},
\blocation{Turku, Finland}
(\byear{1996})
\end{bchapter}
\endbibitem

\bibitem[\protect\citeauthoryear{Urban}{2006}]{urban2006mptp}
\begin{barticle}
\bauthor{\bsnm{Urban}, \binits{J.}}:
\batitle{{MPTP 0.2: Design, implementation, and initial experiments}}.
\bjtitle{Journal of Automated Reasoning}
\bvolume{37},
\bfpage{21}--\blpage{43}
(\byear{2006})
\end{barticle}
\endbibitem

\bibitem[\protect\citeauthoryear{Kaliszyk and
  P{\k{a}}k}{2019}]{kaliszyk2019semantics}
\begin{barticle}
\bauthor{\bsnm{Kaliszyk}, \binits{C.}},
\bauthor{\bsnm{P{\k{a}}k}, \binits{K.}}:
\batitle{{Semantics of Mizar as an Isabelle object logic}}.
\bjtitle{Journal of Automated Reasoning}
\bvolume{63},
\bfpage{557}--\blpage{595}
(\byear{2019})
\end{barticle}
\endbibitem

\bibitem[\protect\citeauthoryear{Urban}{2005}]{xmlizing}
\begin{bchapter}
\bauthor{\bsnm{Urban}, \binits{J.}}:
\bctitle{{XML-izing Mizar: Making Semantic Processing and Presentation of MML
  Easy}}.
In: \beditor{\bsnm{Kohlhase}, \binits{M.}} (ed.)
\bbtitle{Mathematical Knowledge Management, 4th International Conference, {MKM}
  2005, Bremen, Germany, July 15-17, 2005, Revised Selected Papers}.
\bsertitle{Lecture Notes in Computer Science},
vol. \bseriesno{3863},
pp. \bfpage{346}--\blpage{360}.
\bpublisher{Springer}, \blocation{???}
(\byear{2005}).
\doiurl{10.1007/11618027\_23} .
\burl{https://doi.org/10.1007/11618027\_23}
\end{bchapter}
\endbibitem

\bibitem[\protect\citeauthoryear{Bylinski and Alama}{2012}]{bylinski2012}
\begin{bchapter}
\bauthor{\bsnm{Bylinski}, \binits{C.}},
\bauthor{\bsnm{Alama}, \binits{J.}}:
\bctitle{{New Developments in Parsing Mizar}}.
In: \beditor{\bsnm{Jeuring}, \binits{J.}},
\beditor{\bsnm{Campbell}, \binits{J.A.}},
\beditor{\bsnm{Carette}, \binits{J.}},
\beditor{\bsnm{Dos~Reis}, \binits{G.}},
\beditor{\bsnm{Sojka}, \binits{P.}},
\beditor{\bsnm{Wenzel}, \binits{M.}},
\beditor{\bsnm{Sorge}, \binits{V.}} (eds.)
\bbtitle{Intelligent Computer Mathematics},
pp. \bfpage{427}--\blpage{431}.
\bpublisher{Springer},
\blocation{Berlin, Heidelberg}
(\byear{2012})
\end{bchapter}
\endbibitem

\bibitem[\protect\citeauthoryear{Cairns and Gow}{2004}]{cairns200460}
\begin{barticle}
\bauthor{\bsnm{Cairns}, \binits{P.}},
\bauthor{\bsnm{Gow}, \binits{J.}}:
\batitle{Using and parsing the mizar language}.
\bjtitle{Electronic Notes in Theoretical Computer Science}
\bvolume{93},
\bfpage{60}--\blpage{69}
(\byear{2004})
\doiurl{10.1016/j.entcs.2003.12.028} .
\bcomment{Proceedings of the Mathematical Knowledge Management Symposium}
\end{barticle}
\endbibitem

\bibitem[\protect\citeauthoryear{Nakasho}{2019}]{nakasho2019}
\begin{bchapter}
\bauthor{\bsnm{Nakasho}, \binits{K.}}:
\bctitle{{Development of a Flexible Mizar Tokenizer and Parser for Information
  Retrieval System}}.
In: \bbtitle{2019 Federated Conference on Computer Science and Information
  Systems (FedCSIS)},
pp. \bfpage{77}--\blpage{80}
(\byear{2019}).
\doiurl{10.15439/2019F151}
\end{bchapter}
\endbibitem

\bibitem[\protect\citeauthoryear{Adams}{2016}]{adams2016proof}
\begin{barticle}
\bauthor{\bsnm{Adams}, \binits{M.M.}}:
\batitle{Proof auditing formalised mathematics}.
\bjtitle{Journal of Formalized Reasoning}
\bvolume{9}(\bissue{1}),
\bfpage{3}--\blpage{32}
(\byear{2016})
\end{barticle}
\endbibitem

\bibitem[\protect\citeauthoryear{Kun{\v{c}}ar and
  Popescu}{2015}]{kunvcar2015consistent}
\begin{bchapter}
\bauthor{\bsnm{Kun{\v{c}}ar}, \binits{O.}},
\bauthor{\bsnm{Popescu}, \binits{A.}}:
\bctitle{{A consistent foundation for Isabelle/HOL}}.
In: \bbtitle{Interactive Theorem Proving: 6th International Conference, ITP
  2015, Nanjing, China, August 24-27, 2015, Proceedings 6},
pp. \bfpage{234}--\blpage{252}
(\byear{2015}).
\bcomment{Springer}
\end{bchapter}
\endbibitem

\end{thebibliography}

\end{document}